%% file: 2024-ACC-LyapunovNeuralNetworks_finalSubmission.tex
\documentclass[letterpaper, 10 pt, conference]{ieeeconf}

\IEEEoverridecommandlockouts  

\usepackage{algorithm}
\usepackage{algpseudocode}

\input{preamble/common}
\newcommenter{zili}{Blue} 
\newcommenter{sba}{Purple}
\newcommenter{Q}{Red} 
\newcommand{\pe}{^{[p,\varepsilon]}}
\newcommand{\cMpe}{\cM\pe}

\title{\LARGE \bf
Lyapunov Neural Network with Region of Attraction Search
}

\author{Zili Wang$^{}$, Sean B. Andersson$^{}$, and Roberto Tron$^{}$
\thanks{Z.~Wang is with the Division of Systems Engineering, S.B.~Andersson and R.~Tron are with the Division of Systems Engineering and the Department of Mechanical Engineering, Boston University, Boston, MA 02215, USA.
           {\tt\small \{zw2445,sanderss,tron\}@bu.edu}.}%
\thanks{This work was supported in part by NSF FRR-2212051 and the Center for Information and Systems Engineering at Boston University.}
}

\begin{document}

\maketitle
\thispagestyle{empty}
\pagestyle{empty}

\begin{abstract}
   Deep learning methods have been widely used in robotic applications, making learning-enabled control design for complex nonlinear systems a promising direction. Although deep reinforcement learning methods have demonstrated impressive empirical performance, they lack the stability guarantees that are important in safety-critical situations. One way to provide these guarantees is to learn Lyapunov certificates alongside control policies. There are three related problems: 
 \begin{enumerate*}
    \item verify that a given Lyapunov function candidate satisfies the conditions for a given controller on a region,
    \item find a valid Lyapunov function and controller on a given region, and
    \item find a valid Lyapunov function and a controller such that the region of attraction is as large as possible.
  \end{enumerate*} 
  Previous work has shown that if the dynamics are piecewise linear, it is possible to solve problems 1) and 2) by solving a Mixed-Integer Linear Program (MILP). In this work, we build upon this method by proposing a Lyapunov neural network that considers monotonicity over half spaces in different directions. 
  We
  \begin{enumerate*}
    \item propose a specific choice of Lyapunov function architecture that ensures non-negativity and a unique global minimum by construction, and
    \item show that this can be leveraged to find the controller and Lyapunov certificates faster and with a larger valid region by maximizing the size of a square inscribed in a given level set.
  \end{enumerate*}  
 We apply our method to a 2D inverted pendulum, unicycle path following, a 3-D feedback system, and a 4-D cart pole system, and demonstrate it can shorten the training time by half compared to the baseline, as well as find a larger ROA. 
\end{abstract}

\section{Introduction}
When designing controllers for dynamical systems, there are a variety of properties which are important to ensure, including stability, safety, and robustness. These properties can be proven via certificate functions. In this work, we focus on stability, although the methods we develop can apply to other properties as well. A powerful technique for certifying stability of nonlinear systems is through the use of Lyapunov functions \cite{haddad2008nonlinear}. A system's equilibrium point is asymptoically stable in the sense of Lyapunov if, starting from any state within a region (called the region of attraction (ROA)), the future states of the system  eventually remain close to the equilibrium. The main challenge is that finding Lyapunov functions for nonlinear systems is difficult in general and requires intuition/manual verification. 

In this paper, we are interested in automated controller synthesis with Lyapunov stability guarantees. A simple method is to linearize the system around the equilibrium point and then use LQR control, which gives a local stability guarantee~\cite{Khalil2002}. Alternatively, one can use nonlinear polynomial approximation of the dynamics, and use semidefinite programming (SDP) to find a Lyapunov function as a sum-of-square (SOS) polynomial~\cite{ahmadi2016some,Majumdar2012ControlDA}. Recently, with the development and wide application of deep learning in robotics, many approaches proposed neural network representations of more complex Lyapunov functions~\cite{Dawson2022SafeCW} to synthesize controllers, from a simple nonlinear representation\cite{chang2019neural}, to more generalized fully connected neural networks~\cite{DaiH-RSS-21, wu2023neural}. 

In order to synthesize a controller with Lyapunov guarantees, one needs to verify whether the Lyapunov conditions are satisfied for all the states within a region~\cite{liu2021algorithms}. Given a bounded input, the two principal techniques used to verify an input-output relationship of a network are:
\begin{enumerate*}
  \item exact verification using either Satisfiability
  Modulo Theories (SMT)~\cite{chang2019neural} (where quadratic and $\tanh$ activation functions are used), or applying mixed-integer programs (MIP) solvers~\cite{bunel2018unified, DaiH-RSS-21} (where (leaky) ReLU activation is used), and
  \item inexact verification via convex relaxation~\cite{fazlyab2020safety}.
\end{enumerate*}

In this paper, we use $\relu$ neural networks and a MILP verifier to give exact Lyapunov stability guarantees. Typically, the MILP is formulated to verify certain properties of ReLU neural networks such that if the optimal objective value is zero, certain conditions can be certified; otherwise a counter-example will be generated and used to reduce the violation~\cite{wong2018provable}. There are two approaches to train networks to satisfy certain properties: 
\begin{enumerate*}
  \item counter-example guided training that adds counter-examples to the training set and minimizes a surrogate loss function of the Lyapunov condition violation on a training set~\cite{chang2019neural}, and
  \item bi-level optimization that minimizes the maximum violation of the Lyapunov conditions via gradient descent~\cite{DaiH-RSS-21}.
\end{enumerate*}  
We use a bi-level optimization training process similar to \cite{DaiH-RSS-21} due to its observed better performance in higher dimensional systems. 

In safety critical systems, it is particularly desirable to understand and find the largest ROA. 
There are several approaches to maximize the ROA. In counter-example guided training, a cost term between the $L_2$ norm of the state and the Lyapunov function can be added to the surrogate Lyapunov loss function to regulate how quickly the Lyapunov function value increases with respect to the radius of the level sets~\cite{chang2019neural,zhou2022neural}; or the Lyapunov function and the sub-level set (inner approximation of the ROA) can be searched simultaneously~\cite{Richards2018TheLN,Wei2023NeuralLC}. However, there is currently no existing work based on the bi-level optimization training procedure to find the largest ROA. 

\myparagraph{Paper contributions}
Our paper contains several contributions on the theoretical side:
\begin{itemize}
  \item While there are existing methods to construct Lyapunov neural networks that give positive definiteness\cite{Richards2018TheLN,chen2021learning_roa, Wei2023NeuralLC}, they are in quadratic form and cannot be verified by MILPs. Instead, we provide a Lyapunov $\relu$ network architecture that is composed of monotonic functions over half-spaces in different directions. The architecture is positive definite and has a unique global minimum. 
  \item The above architecture of the Lyapunov neural network enables a novel approach for finding a controller and its corresponding Lyapunov function with largest possible ROA. Since our architecture has nested level sets, we can enlarge the ROA by enlarging an inscribed cube as a MILP.
  The maximization of the ROA is done using gradient descent, with the gradient from MILP.
\end{itemize}


Through simulations, we provide the following key results:
\begin{itemize}
  \item Compared to \cite{DaiH-RSS-21}, we shorten the training time by half. We attribute this to our model structure, which effectively reduces the search space for the Lyapunov function. 
  \item We demonstrate the strength of our approach in expanding the region of attraction in complex 2D to 4D systems. Our result shows larger/more predictable ROA than previous work.
\end{itemize}

\section{Preliminaries}\label{sec:preliminaries}
\subsection{Lyapunov's Direct Method for Stability}
We briefly review Lyapunov stability theory \cite{Khalil2002}.


\begin{definition}[Control Lyapunov Functions]\label{def:control-lyapunov-conditions}
  Consider the nonlinear system:
  \begin{equation}\label{eq:dicrete_system}
    x_{t+1} = f(x_t, u_t),
  \end{equation}
  where $x_t$ is a state in a domain $\mathcal{X}\subset \mathbb{R}^{n_x}$ and $u_t$ is the control in a input space $\mathcal{U}\subset \mathbb{R}^{n_u}$. Let $\pi(\cdot): \mathbb{R}^{n_x} \rightarrow \mathbb{R}^{n_u}$ denote a control policy with $u_t = \pi(x_t)$. An exponentially-stable control-Lyapunov function (CLF) is a function $V:\mathbb{R}^{n_x}\rightarrow \mathbb{R}$ that is continuously-differentiable and satisfies 
  \begin{subequations}\label{eq:lyapunov_conditions}
    \begin{align}
      V(0) = 0, &\label{eq:lyapunov_conditions_a}\\
      V(x)>0, ~\forall x\in \mathcal{R}\setminus\{0\}, &\label{eq:lyapunov_conditions_b}\\
      V\biggl(f\bigl(x,\pi(x)\bigl)\biggl)-V(x)\leq -\epsilon V(x), ~\forall x\in \mathcal{R}\setminus\{0\},& \label{eq:lyapunov_conditions_c}
    \end{align}
\end{subequations}
over a subset $\mathcal{R}\subset \mathcal{X}$.
\end{definition}

\subsection{ReLUs, MILP, Verification, and Learning}\label{subsec:relu_mip}
We review here how to represent a $\relu$ network and perform verfication over a MILP.

A Leaky ReLU (Rectified Linear Unit) function has the form
\begin{equation}\label{eq:relu_function}
  \sigma(y) = \max(0, cy),
\end{equation}
where $0\leq c<1$ is a given positive scalar. For a fully-connected neural network with leaky $\relu$ activation, the input/output relationship can be represented as
\begin{equation}\label{eq:relu_network}
  \begin{aligned}
      \phi(x;\theta) &= \hat{z}_{k-1}\\
      \hat{z_i} &= W_i z_{i-1} + b_i, ~~i=1,2,\dots,k-1, \\
        z_i &= \sigma(\hat{z_i})=\max(\hat{z_i},c\hat{z}_i), ~~i=1,2,\dots,k-1, \\
        z_0 &= x,
      \end{aligned}
    \end{equation}
where $k$ is the number of layers and $\theta$ contains all the network parameters (i.e. all $W_i$ and $b_i$). 

The MILP formulation stems from the only nonlinear part in the network, namely, the $\relu$ activation function. 
\begin{remark}\label{remark:relu_mip}
If the input to $\sigma(\cdot)$ is bounded ($\hat{z}_{i,lo}\leq \hat{z}_i \leq \hat{z}_{i,up}$), the input/output relationship through $\relu$ activation can be encoded as mixed-integer linear constraints
\begin{equation}\label{eq:relu_milp}
  \begin{aligned}
    z_i\geq \hat{z}_i, ~z_i\geq c\hat{z}_i, &\\
    z_i\leq c\hat{z}_i-(c-1)\hat{z}_{i,up}\beta,&\\ 
    z_i\leq \hat{z}_i-(c-1)\hat{z}_{i,lo}(\beta-1), &\\
    \beta \in \{0,1\}, &
  \end{aligned}
\end{equation}
where the binary variable $\beta$ is active when $\hat{z}_i\geq0$. 
\end{remark}

The $\relu$ network \eqref{eq:relu_network} can be represented as a MILP by replacing the activation units with \eqref{eq:relu_milp}. In a multiple-layer network, the continuous variables in the MILP correspond to the input and each layer's output. With a bounded neural network input, the bound of each $\relu$ neuron input can be computed from the bounds of the connected nodes from the previous layers; these can be obtained by using Interval Arithmetic \cite{wong2018provable}, solving a Linear Programming \cite{tjeng2017evaluating}, or solving a MILP problem \cite{fischetti2018deep}. The relationship between network input and output is therefore fully captured in mixed-integer linear constraints. Mixed Integer Programs have been widely applied in domains including path planning~\cite{ioan2021mixed}, temporal logical~\cite{raman2014model}, and controller synthesis~\cite{DaiH-RSS-21,chen2021learning} to solve different problems.

\section{Problem Statement}
We consider a discrete-time nonlinear dynamical system \eqref{eq:dicrete_system}
with bounded state space $\mathcal{D}\subset \mathcal{X}$ and control space $\mathcal{U}$. We assume that the dynamical system is represented by a piecewise linear function (however, as will be shown in Sec.~\ref{Sec:simulations}, in practice, our techniques are also applicable to $\relu$ neural network approximation of general systems~\cite{leshno1993multilayer}). In the following, we consider three problems: 

\myparagraph{Problem I} Verify the satisfaction of the Lyapunov conditions for a given controller $\pi(\cdot)$ and Lyapunov function over a region $\mathcal{D}$.

\myparagraph{Problem II} Find an exponentially stable controller $\pi(\cdot)$ and Lyapunov function  on $\mathcal{D}$.

\myparagraph{Problem III} Find a controller $\pi(\cdot)$ and Lyapunov function such that the ROA is as large as possible.

We use the Lyapunov conditions~\eqref{eq:lyapunov_conditions} to imply the system stability. However, the Lyapunov conditions themselves do not directly ensure safety, since the system can deviate arbitrarily far before eventually settling to the stable equilibrium. The region of attraction, however, defines a forward invariant set that is guaranteed to contain all possible trajectories of the system, and therefore establishes safety. 
A typical under-approximation of the ROA can be found by verifying the Lyapunov conditions on a sub-level set of the Lyapunov function $\mathcal{R}_\rho=\{x|V(x)\leq\rho\}$ for some value of the parameter $\rho>0$.

In our setting, the dynamic system is represented by a fixed neural network $\phi_{dyn}$ with leaky $\relu$ activation functions, 
\begin{equation}\label{eq:dynamics}
  f(x,u) = \phi_{dyn}(x,u)-\phi_{dyn}(x_{eq},u_{eq}) + x_{eq},
\end{equation}
where $x_{eq}$, $u_{eq}$ are the state and control at the equilibrium point. The equation guarantees that at the equilibrium state and control, the next state remains at the equilibrium point. 

The control policy is represented by a feedforward, fully connected neural network $\phi_{\pi}$ with leaky $\relu$ activation,
\begin{equation}
  \pi(x) = \text{clamp}(\phi_\pi(x)-\phi_\pi(x_{eq})+u_{eq},u_{min}, u_{max}),
\end{equation}
where $u_{min}$, $u_{max}$ are the lower and upper control limit, and $\text{clamp}(\cdot)$ clamps (element-wise) all the values within the control limits.
 
The Lyapunov function is represented by another neural network $\varphi_{v}(x)$,
\begin{equation}\label{eq:lyapunov_with_R}
  V(x) = \phi_{v}(x) + \lambda |R(x-x_{eq})|_1
\end{equation}
where $R$ is a full rank matrix and $\lambda$ is a hyperparameter that balances $\phi_{v}(x)$ and the term containing $R$. The term containing $R$ is strictly positive everywhere except at $x_{eq}$. In practice, adding this term facilitates learning, but it is not strictly necessary for our guarantees on $V(\cdot)$.  

\section{Lyapunov Neural Network}

In the section, we introduce a Lyapunov Neural Network that can be built from $\relu$ units and has a single global minimum and no other global minima \emph{by construction}. With this architecture, some Lyapunov conditions are automatically satisfied, allowing us to formulate a simpler optimization problem to verify remaining conditions to synthesize the controller and Lyapunov function.

\subsection{Monotonic Layers}
The Lyapunov Neural Network is built from \emph{units} that represent monotonically increasing functions; and \emph{layers} that define \emph{star convex} functions and are guaranteed to have a unique global minimum.
  

\begin{definition}[Monotone functions]
  A function $m(\cdot): \real{}$ $\rightarrow \real{}$ is of class $\cM$ if it has the following properties:
  \begin{itemize}
  \item $m(\cdot)$ is continuous, 
  \item $m(0)=0$,
  \item $m(\cdot)$ is strictly increasing everywhere. \label{monotone function condition3}
  \end{itemize}
   If the last condition holds only for non-negative arguments, $m(\cdot)$ is of class $\mathcal{M}_+$.
  
  A function $m(\cdot)$ is of class $\cMpe$ if it is of class $\cM$ and:
  \begin{itemize}
  \item it is piecewise linear,
  \item its derivative has $p\in\naturals{}$ discontinuity points,
  \item all defined derivatives $m'(y)$ are greater than $\varepsilon$, i.e., $\varepsilon$ is a lower bound on the slope of $m$.
  \end{itemize}
\end{definition}
If the function is of class $\cM_+$ and the above conditions hold, then it is of class $\cMpe_+$. Note that functions of class $\cM$ and $\cM_+$ are of class $\cK$~\cite{Khalil2002} when restricted on $[0,\infty)$.


\begin{definition}[Monotonic unit]
  A \emph{monotonic unit} $m\pe$ of order $p$ is defined as 
  \begin{equation}\label{eq:m y}
    m(y)=a\transpose\relu(y\vct{1}-b),
  \end{equation}
where $a,b\in\real{p}$, $\vct{1}\in\real{p}$ is the vector of all ones, and $\relu$ is the element-wise Rectifier Linear Unit operator \eqref{eq:relu_function}.
\end{definition}

\begin{proposition}\label{proposition:monotonic_unit}
  A monotonic unit $m$ is of class $\cMpe$ if:
  \begin{enumerate}
  \item $b\in\real{p}$ is a vector of \emph{biases} satisfying
    $[b]_{i+1}<[b]_i$
    where $[\cdot]_j$ is the $j$-th entry of a vector.
  \item $a\in\real{p}$ is a vector of \emph{slopes} satisfying
    \begin{enumerate}[label=(\alph*)]
    \item \label{it:a one}$[a]_1>0$,
    \item\label{it:a sum} $[a]_{i+1}\geq -\sum_{i'=1}^{i-1} a_{i'}+\varepsilon$.
    \end{enumerate}
  \end{enumerate}
  A monotonic unit $m$ is of class $\cMpe_+$ if the additive condition $[b]_1=0$ is satisfied.
\end{proposition}
\begin{proof} Expanding the inner product in \eqref{eq:m y}, and using the fact that $\relu(y-b_i)=0$ if $b_i\geq y$, we have 
\begin{equation}\label{eq:m y prime}
m(y)=\sum_{i=1}^{i'} a_i y-\sum_i^{i'} a_i b_i,
\end{equation}
where $i'$ is the largest $i$ such that $b_i<y$, that is, $m$ is piecewise linear in $y$. 
Condition \ref{it:a sum} implies $\sum_{i=1}^{i''+1} a_{i}\geq \varepsilon$ for any $i''$; combined with \ref{it:a one} implies the slope in \eqref{eq:m y prime} is always no less than $\varepsilon$, i.e., the function is always monotonically increasing and hence of class $\cMpe$. If $[b]_1=0$, then $m(y)=0$ for all $y\leq 0$, and the function is of class $\cMpe_+$.
\end{proof}

\begin{definition}[Monotonic Layer]
A \emph{monotonic layer} $M:\real{d_{in}}\to\real{d_{out}}$ is defined as a combination of monotonic units:
\begin{equation}\label{eq:M layer}
[M]_j(x)=\sum_{i=1}^{n_M} c_i m_i(v_i\transpose x)
\end{equation}
where $x\in\real{n_x}$ is the input, 
$v_i\in\real{n_x}$ is a direction and $v_i\neq 0$, $c_i>0$, and $n_M\geq d+1$. 
\end{definition}

\begin{lemma}
   For a given direction $v_i$, the function $m_i(v_i\transpose \cdot):\real{n_x}\rightarrow \real{}$ is:
  \begin{enumerate}
      \item monotonically increasing in the direction $v_i$,
      \item constant along any direction orthogonal to $v_i$.
  \end{enumerate}
\end{lemma}
The proof is omitted to save space. A monotonic layer is monotonically increasing along radial lines $x(t)=vt$, where $v\in\real{n_x}$. An example of a monotonic layer in $\cMpe_+$ is shown in Fig.~\ref{fig:monotonic_demo}. 
\begin{figure}[htbp!]
  \centering
  \subfloat[Layer\label{fig:monotonic_layer}]{%
  \includegraphics[width=0.3\linewidth]{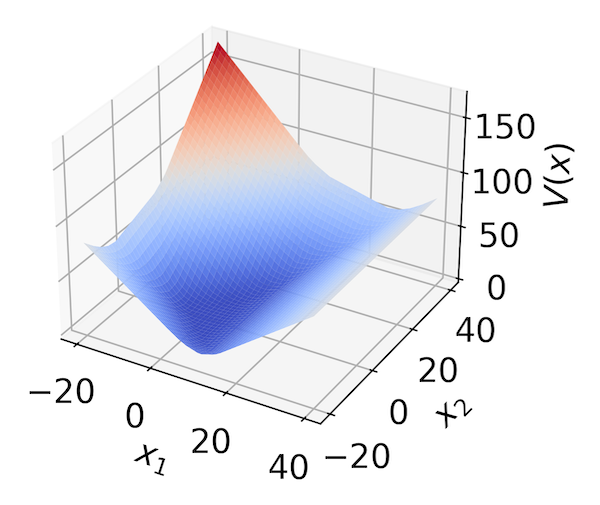}%
}
  \subfloat[Monotonic units in $\cMpe_+$\label{fig:monotonic_units}]{%
  \includegraphics[width=0.7\linewidth]{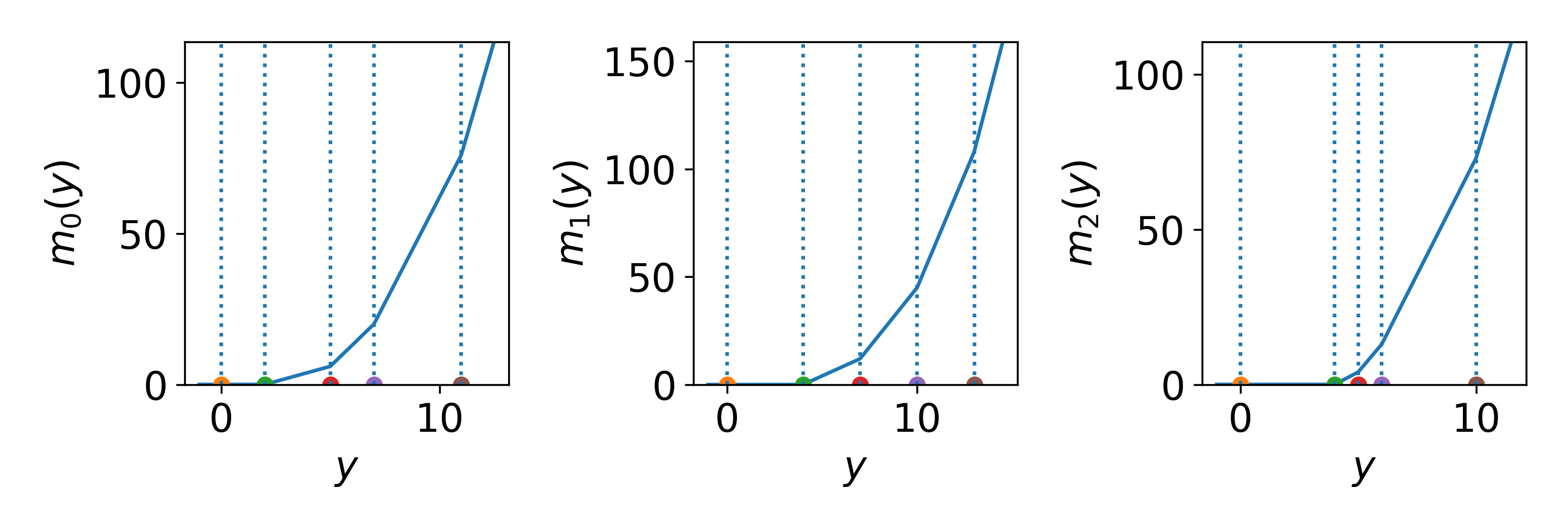}%
}
  \caption{Illustration of Monotonic Layer $V(x)=[M]_1=\sum_{i=1}^{3}m_i(v_i^T x)$ with $v_i\in\{[1;0],[-1;1],[-1;-1]\}$, each $m_i(\cdot)$ is shown in (b).}
    \label{fig:monotonic_demo}
\end{figure}

\subsection{Proposed Lyapunov Network}

\begin{definition}[Lyapunov Neural Network]
An $n$-layer \emph{Lya} \emph{-punov Neural Network} $\phi_v:\real{n_x}\to\real{}$ in \eqref{eq:lyapunov_with_R} is defined as the composition of monotonic layers:
\begin{equation*}
\phi_v(x)=(M_n\circ M_{n-1} \circ \ldots \circ M_1)(x)
\end{equation*}
\end{definition}

\begin{proposition}
  Assuming $\vct{0}$ is in the convex hull of a set of directions $\{v_i^M\}$ for $M_1$, $\phi_v(x)$ has the following properties: 
    \begin{itemize}
    \item $\phi_v(x)\geq 0$, $V(0)=0 \iff x=0$,
    \item $\phi_v(x)$ is monotonically increasing along radial lines,
    \item $\phi_v(x)$ is radially unbounded, has a unique global minimum at $x=0$, and no other local minima.
    \end{itemize}
\end{proposition}
\begin{proof}
  By definition, $M_i$ is a function that is monotonically increasing along radial lines originating from zero. The composition of multiple such functions inherits this property. Since the monotonic units are of class $\cMpe$ and $\vct{0}\in \cvxhull(\{v_i^M\})$ ($\cvxhull(\cdot)$ denotes the convex hull of a set), $\phi_v(x)$ is increasing everywhere from zero and radially unbounded. Therefore, it has a unique global minimum at $\phi_v(0)=0$. 
\end{proof} 
\begin{remark}
  The $R$ term in \eqref{eq:lyapunov_with_R}
  can be written in the form \eqref{eq:M layer}. Under this representation, the directions $\{v_i^R\}$ for the $R$ term satisfy the condition $\vct{0}\in\cvxhull(\{v_i^R\})$. Hence, with the $R$ term, $V(x)$ can be proved to possess the same properties as $\phi_v(x)$, but without requiring that $\vct{0} \in \cvxhull(\{v_i^M\})$. 
\end{remark}

\begin{corollary}
  $\phi_v(x)$ is star convex with respect to the origin and has nesting, compact, simply connected, star-convex (with respect to the origin) level sets (i.e., $V\inverse(v_1)\subset V\inverse(v_2)$ for any $v_1<v_2$).
\end{corollary}
\begin{proof}
The function from the origin to any point $x$ is monotonically increasing, hence every point between $0$ and $x$ is inside any level set containing $x$. The claims follow.
\end{proof}
\begin{remark}
  The bias and slope vectors in the monotonic units can be converted to $\relu$ network parameters, in a way similar to~\cite{wang2022bearing}. To satisfy the monotonic unit properties in Proposition~\ref{proposition:monotonic_unit}, the neural network parameters are thresholded to be positive/non-negative during training. 
\end{remark}

\subsection{Verification of Lyapunov Conditions}\label{subsec:bounded_MIP}
With the proposed structure on $V(x)$, Lyapunov conditions \eqref{eq:lyapunov_conditions_a} and \eqref{eq:lyapunov_conditions_b} are satisfied by construction.

To ensure the stability of a given controller, one needs to verify the Lyapunov condition~\eqref{eq:lyapunov_conditions_c} in a given bounded polytope $\mathcal{D}$ around the equilibrium. The verification problem (Problem I) can be formulated as
\begin{equation}\label{eq:verifier_problem}
  \max_{x\in \mathcal{D}} V\biggl(f\bigl(x,\pi(x)\bigl)\biggl)-V(x) + \epsilon V(x).
\end{equation}
In our setting,  $V(x)$, $f(x,\pi(x))$ and $V(f(x,\pi(x)))$ are piecewise-affine functions of $x$; therefore, the problem can be solved through MILPs as in Sec.~\ref{subsec:relu_mip}. If the optimal value is zero, the system is certified; otherwise the optimal value is greater than $0$ and the optimal point is an counter-example.

\subsection{Learning Procedure}
With the contents in Sec.~\ref{sec:preliminaries}, the problem of finding a valid controller and Lyapunov function (Problem II) becomes the problem of minimizing the maximum violation of Lyapunov condition~\eqref{eq:lyapunov_conditions_c}. This can be formulated as the bi-level min-max optimization:
\begin{align}\label{eq:set_training}
  \min_{\theta} \left(\max_{x\in \mathcal{D}} V(f(x,\pi(x))) - V(x) + \epsilon V(x)\right),
\end{align}
where $\theta$ contains the trainable controller and Lyapunov network parameters.

\begin{algorithm}[t]
  \caption{Train controller/Lyapunov function through min-max optimization}\label{alg:lyapTraining_withSet}
  \label{algo:verfication_in_set}
  \begin{algorithmic}[1]
  \State {Given: a candidate control policy $\pi$ and a candidate Lyapunov function $V$}.
  \While{not converged}
  \State {Solve MILP $\max_{x\in \mathcal{D}} V(f(x,\pi(x))) - V(x) + \epsilon V(x)$}
  \If{The MILP has maximal objective $>0$}
      \State Compute the gradient of the MILP objectives w.r.t network parameters $\theta$ as $\frac{\partial \gamma}{\partial \theta}$.
      \State $\theta = \theta - \text{StepSize}\times\frac{\partial \gamma}{\partial \theta}$
  \EndIf
  \EndWhile
  \end{algorithmic}
  \end{algorithm}

Such a problem can be solved in an iterative procedure \cite{wong2018provable} as shown in Algorithm.~\ref{algo:verfication_in_set}. At each iteration, 
with fixed network parameters, the inner maximization problem can be solved using MILP solvers. A general MILP formulation is
\begin{equation}
  \begin{aligned}
    \gamma(\theta) = &\max_{s,\beta} a_\theta^T s + b_\theta \beta \\
                    &\text{s.t}~A_\theta s + B_\theta \beta \leq c_\theta,
    \end{aligned}
  \end{equation}
  where $a_\theta, b_\theta, A_\theta, B_\theta, c_\theta$ are explicit functions of $\theta$, $\beta$ are binary variables, and $s$ (slack variables and $x$) are all the continuous variables in the MILP related to \eqref{eq:relu_milp}. 

  After solving the MILP to optimality, with optimal $s^*$ and  $\beta^*$, the active linear constraints at the solutions are $A^{act}_\theta s^* + B^{act}_\theta \beta^* = c^{act}_\theta$~\cite{DaiH-RSS-21}. Assuming an infinitesimal
  change of $\theta$ does not change the binary variable solution nor the indices of active constraint, and letting $(A^{act}_\theta)^\dagger$ be the pseudo-inverse of $A^{act}_\theta$, then the gradient of the MILP optimal objective with respect to the network parameters, given by
  \begin{equation}\label{eq:max_milp_optimal}
    \gamma(\theta) = a_\theta^T (A^{act}_\theta)^\dagger (c^{act}_\theta - B^{act}_\theta \beta^*)+ b_\theta^T \beta^*,
  \end{equation}
   can be used to update the network. In the case where the assumption does not hold, the problem is degenerate, but the gradient in \eqref{eq:max_milp_optimal} can still be used to improve the network. For further discussion, see \cite{DaiH-RSS-21}. 

  The procedure terminates when the inner maximization objective gives a zero value, yielding a valid controller and a corresponding Lyapunov function that solves Problem II. 
  


\section{Region of Attraction Expansion}
One of the main benefit of our proposed network is that, by leveraging the special monotonicity structure, we can offer a method to enlarge the ROA. We first introduce a bi-level optimization simliar to \eqref{eq:set_training} but where $\mathcal{D}$ is a fixed sub-level set (used to represent the ROA). 
We then expand the level set by defining an inscribed square, and maximizing its area. If the expanded level set satisfies all the Lyapunov conditions, then it becomes the new ROA. This iterative procedure addressing Problem III is summarized in Algorithm~\ref{alg:lyapExpansion_withSet}.
\begin{algorithm}[t]
  \caption{Expand Region of Attraction with Lyapunov}\label{alg:lyapExpansion_withSet}
  \begin{algorithmic}[1]
  \State {Given: a candidate control policy $\pi$ and a candidate Lyapunov function $V$}.
  \While{$\mathcal{R}_r \in \mathcal{D}$} 
  \State Find a valid controller in  $\mathcal{R}_r$ using modified Algorithm~\ref{algo:verfication_in_set} (Sec.~\ref{subsec:level_set_verification})
  \While{less than maximum trials}
      \State {Solve MILP problem \eqref{eq:g_l} to find the inscribed square (Sec.~\ref{subsec:level_set_expansion})}
      \State Compute the gradient of the $l^*(\theta)$ w.r.t. network parameters using $\frac{\partial l^*}{\partial \theta}$
      \State Select a moving direction from \eqref{eq:optimizer_direction} and implement gradient descent (Sec.~\ref{subsec:expand_direction})
  \EndWhile
  \EndWhile
  \end{algorithmic}
  \end{algorithm}

\subsection{Training over a Sub-level Set}\label{subsec:level_set_verification}
 Instead of finding the maximum violation over a fixed pre-defined bounded region $\mathcal{D}$ as in Sec.~\ref{subsec:bounded_MIP}, we are interested in a compact sub-level set $\mathcal{R}_r = \{x|V(x)\leq r\}$, where $r$ is pre-defined and the region of ${R}_r$ depends on $\theta$. The problem is a modified version of Problem II, where the region is given as a level set of the Lyapunov function itself, and a simple modification of Algorithm.~\ref{alg:lyapTraining_withSet} can be used to solve it. The problem of finding a valid controller and Lyapunov function over a level set then becomes
\begin{align}\label{eq:levelSet_training}
  \min_{\theta} \left(\max_{x\in \mathcal{B}_r, V(x)\leq r} V(f(x,\pi(x))) - V(x) + \epsilon V(x)\right)
\end{align}

One difficulty is that, as stated in Remark~\ref{remark:relu_mip}, in order to provide the MILP formulation of the problem, one needs to know the bounded region $\mathcal{B}_r$ that contains ${R}_r$ (i.e. ${R}_r\subset \mathcal{B}_r$ ). To find $\mathcal{B}_r$, we make one additional assumption.
\begin{assumption}
  We assume the monotonic direction set $\{v_i\}$ contains all the axes (with indices set $\mathcal{I}_{axes}$).
\end{assumption} 

\begin{proposition}\label{proposition:bounded_domain_fromSet}
  Given a network $\phi_v(x)$ represented as a combination of monotonic functions as $\sum_{i=1}^{n_M} c_i a_i^T \relu(v_i^T x-b_i)$, the bounded domain $\mathcal{B}_r$ containing the sub-level set (i.e. $\mathcal{R}_r \subseteq  \mathcal{B}_r$) can be found as 
\begin{equation}\label{eq:rectSet_r}
  \begin{aligned}
    p_i &= m_i^{-1} (r), \forall i\in \mathcal{I}_{axes} \\
    \vct{p} &= \text{stack} \{ p_i\} \\
    [\mathcal{B}_r]_j &= ||[\vct{p}]_j||_{\infty}, ~j\leq n_x
  \end{aligned}
\end{equation}
where $\vct{p}\in\mathbb{R}^{n_M\times n_x}$ and $[\mathcal{B}_r]_j$ denotes the largest value of $\vct{p}$ in $j_{th}$ dimension.
\end{proposition}
\begin{figure}[b]
  \centering
  \includegraphics[width=0.75\linewidth]{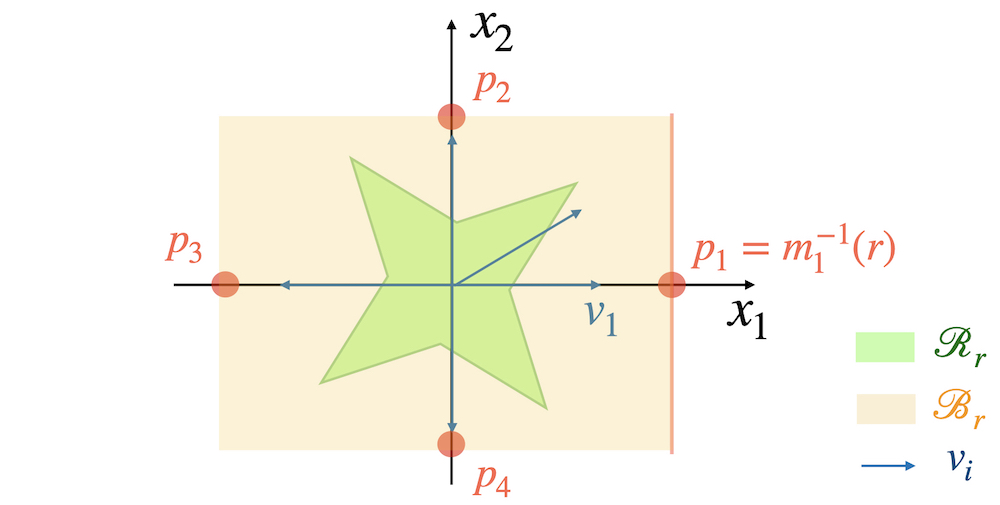}%

  \caption{Visualization for the proof of proposition \ref{proposition:bounded_domain_fromSet}.}
    \label{fig:propostion4_proof}
\end{figure}
\begin{proof}
   Without loss of generality, choose axis $i=1$, and define $p_1$ as the point where $m_1(v_1^T p_1) = r$. Since $V(x) = \sum{m_{i'}(x)}$ and $m_{i'}(x)\geq 0$ for any $i'$, $p_1$ is outside of the sub-level set $V\inverse(r)$, i.e. $V(p_1)>r$.
   Furthermore, $m_1(v_1^T x)$ is constant along hyperplanes normal to $v_1$. Hence, any point on the plane $x_1  = [p_1]_1$ has value no smaller than $r$, and the plane is also outside of $V\inverse(r)$ (see Fig.~\ref{fig:propostion4_proof}). Repeating the same process for all directions $i$, we have that the sub-level set is contained by the set surrounded by all the hyperplanes, i.e., the bounding box that contains $\mathcal{R}_r$.
\end{proof}


\subsection{Approximating a Sub-level Set}\label{subsec:level_set_expansion}
The problem of finding the inscribed cube can be formulated as a MILP:


\begin{equation} \label{eq:l-inf norm of set}
    \min_{l} l \quad \textrm{s.t.} \quad  V(x)\leq l||x-x_{eq}||_\infty,~ \forall x\in\mathcal{R}_r.
\end{equation}

The $l_{\infty}$ norm term can be formulated as mixed-integer linear constraints. Letting 
\begin{equation}\label{eq:g_l}
g(l) = \max_{x\in \mathcal{R}_r,~||x-x_{eq}||_\infty\geq\varepsilon} V(x) - l||x-x_{eq}||_\infty,
\end{equation}
where $\varepsilon$ is a small constant, \eqref{eq:l-inf norm of set} becomes
  \begin{equation}\label{eq:linf_optimization_problem}
  \min_{l}\quad l \quad \textrm{s.t.} \quad g(l) \leq 0
    \end{equation}
 From nonlinear optimization theory \cite{bertsekas1997nonlinear}, the optimal solution of \eqref{eq:linf_optimization_problem} is achieved when the constraints are active (i.e., $g(l^*) = 0$); if not, $l^*$ would be unbounded and we could reach any arbitrarily low value. Therefore, the optimal objective $l^*$ can be found by finding the root of \eqref{eq:g_l} using a numerical method such as bisection search~\cite{burden2015numerical}.

\subsection{Selecting the ROA Expansion Direction}\label{subsec:expand_direction}
\begin{figure}[b]
  \centering
  \includegraphics[width=0.7\linewidth]{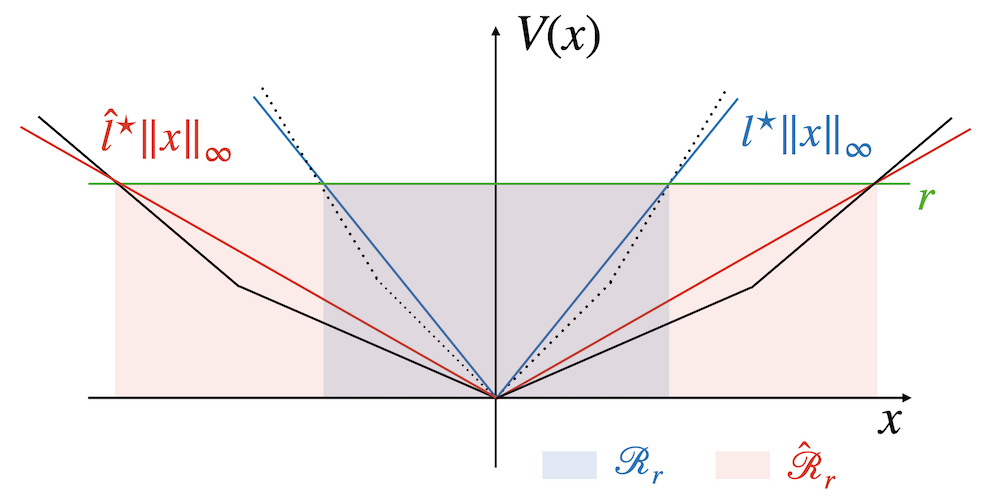}
  \caption{Illustration of level set expansion on 1D Lyapunov functions: the blue box (identified with $l^\star$) is the original level set and the red one (identified with $\hat{l}^\star$) is after expansion.}
    \label{fig:1d_roa_expand_vis}
  \end{figure}
As illustrated in Fig.~\ref{fig:1d_roa_expand_vis}, the sub-level set $\mathcal{R}_r$ can be enlarged by minimizing $l^*$ (which is a function of $\theta$). Specifically, the relationship between $l^*$ and the neural network parameters can be obtained after solving \eqref{eq:g_l} to optimality (i.e. $g(l)=0$). By fixing all the binary variables to their optimal solutions, and keeping only the active linear constraints, \eqref{eq:g_l} becomes
\begin{equation}\label{eq:inner_Lp_findL}
    \xi(\theta) = ~c_{\theta,v}^T s_v  + d_{\theta,v} - l_\theta^* s_{l_{inf}}
\end{equation}
where $s_v$ and $s_{l_{inf}}$ need to satisfy the equations
\begin{equation}
 A_{\theta,v} s_v = b_{\theta,v},
              ~~A_{\theta,l_{inf}} s_{l_{inf}} = b_{\theta,{l_{inf}}}. 
\end{equation}
Here, $s$ contains all the continuous variables in the MILP (including $x$ and other slack variables), and the subscript $v$ ($l_{inf}$) indicates the variable is related to $V(x)$ ($l_{\infty}$). From \eqref{eq:inner_Lp_findL} we have
\begin{equation} \label{eq:l_mip_solution}
  l^*= \frac{c_{\theta,v} s_v + d_{\theta,v}}{s_{l_{inf}}}
  = \frac{c_{\theta,v} (A_{\theta,v}^{-1} b_{\theta,v}) + d_{\theta,v}}{(A_{\theta,{l_{inf}}^{-1}} b_{\theta,{l_{inf}}})}.
\end{equation}
We can use gradient descent on $l^*$ with respect to the network parameters $\theta$ to enlarge the area of the inscribed square, and hence the area of $\mathcal{R}_r$. However, moving directly in the direction of the gradient might lead to a new $\mathcal{R}_r$ that severely violates the Lyapunov conditions. Therefore, a gradient descent direction $w$ should satisfy
  \begin{equation}
      \nabla_\theta \gamma(\theta) w \geq 0,~
      \nabla_\theta l^*(\theta) w \geq 0.
    \end{equation}
An optimal selection can be formulated as an LP problem
\begin{equation}\label{eq:optimizer_direction}
  \begin{aligned}
  \max_{w} ~&\nabla_\theta l^*(\theta) w +s \\
  \text{s.t.}~&\nabla_\theta \gamma(\theta) w \geq s,~ \nabla_\theta l^*(\theta) w \geq 0, \\
    &||w||\leq 1,~s \geq 0.
  \end{aligned}
  \end{equation}

\section{Simulation}\label{Sec:simulations}
\begin{table*}[]
  \caption{Network Structure and Learning Time: numbers in brackets in the baseline network represent the number of neurons in each layer. Our network has a number of directions (dirs) that define the monotonic functions. Each monotonic function is a piecewise linear function with a number of pieces (pcs). The unit for training time is in minutes.
   }
  \label{table:example_networks}
  \centering
  \begin{tabular}{lcccc|cc} 
          \toprule
          &\multicolumn{1}{c}{Baseline Network}
          &\multicolumn{1}{c}{Training Time}
          &\multicolumn{1}{c}{Our Network}
          &\multicolumn{1}{c}{Training Time}
          &\multicolumn{1}{|c}{ROA Training Time} \\
          \midrule
          Inverted Pendulum  & $[8,8,6,1]$ &  97.2 & [5 dirs/4 pcs] & 24.3 &250.0\\
     
          Path Following  
          & $[8,8,6,1]$ & 45.1 & [6 dirs/4 pcs] & 4.8 & 32.4\\

          Third-Order Strict  
          & $[8,8,6,1]$ & 21.4 & [7 dirs/4 pcs] & 10.1 &76.9\\

          Cart Pole 
          & $[8,7,5,1]$ & 74.8 & [7 dirs/4 pcs] &57.2 &69.3\\
          \bottomrule
  \end{tabular}
\end{table*}


In this section, we demonstrate \begin{enumerate*}\item our proposed monotonic network by comparing the convergence time with a previous approach that used simple $\relu$ networks to represent the Lyapunov function, \item our proposed searching method to find the largest possible ROA\end{enumerate*}, by applying it to four examples: an inverted pendulum, a unicycle path following system, a third-order strict-feedback system, and a cart pole system.

The dynamic model used for learning in each case is a fixed network that was pre-trained on a sufficiently large state space $\mathcal{X}$ to approximate the real dynamical system. Both the controller and Lyapunov networks were initialized by training them to approximate the LQR solution. The initial $R$ term in \eqref{eq:lyapunov_with_R} was composed of slight deviations of the eigenvectors of the solution to the Riccati equation of the LQR problem while ensuring the full rank property. All networks were trained following the approach in~\cite{DaiH-RSS-21}. The directions defining the monotonic units in the Lyapunov Network were fixed as a combination of the axis vectors and selected eigenvectors of the solution to the Riccati equation. To verify a large region, we initially verified a small region and gradually enlarged that region. This gave more stable convergence of the algorithms in practice. Note that for running the actual simulations in the results below, the full dynamic model of each system, not its network-based approximation, was used.

\subsection{Systems \label{subsec:systems}}
\myparagraph{System I: Stationary Inverted Pendulum} A second-order nonlinear system with dynamics 
\begin{equation}
  \ddot{\theta} = \frac{u-mgl\sin(\theta)-b\dot{\theta}}{ml^2},
\end{equation}
where $\theta=\pi$ is the upward position of the pendulum and $\theta$ is positive in the counter-clockwise direction. The stabilizing equilibrium state is $\theta=\pi$, $\dot{\theta}=0$. We set the mass term to $m=1$ and the damping coefficient to $b=0.1$. The states are $(\theta,\dot{\theta})$ and the state space is 
\begin{equation*}
  \mathcal{D} = \{(\theta,\dot{\theta} ) | 0\leq \theta\leq 2\pi,
  \lvert \dot{\theta} \rvert \leq 5 \}.
\end{equation*}
The control space is $\{u | \lvert u \rvert \leq 10\}$.

\myparagraph{System II: Wheeled Vehicle Path Following on a Unit Circle} A curvature-dependent unicycle system with kinematics
\begin{equation}
  \begin{pmatrix}
  \dot{d_e}\\[\jot]
  \dot{\theta_e}
  \end{pmatrix}=\begin{pmatrix}
   v \sin{\theta_e}\\[\jot]u-\frac{v \kappa(s)\cos{\theta_e}}{1-d_e \kappa(s)}.
  \end{pmatrix} 
  \end{equation}
We set the curvature to $\kappa(s)=1$ (defining a unit circle), and fix $v=6$. The stabilizing equilibrium state is $d_e = 0$, $\theta_e = 0$, at which the system has no tracking error. The states are $(d_e, \theta_e)$ and the state space is 
\begin{equation*}
  \mathcal{D} = \{(\theta_e, d_e) | \lvert \theta_e \rvert \leq 0.8,
  \lvert d  _e \rvert \leq 0.8 \}.
\end{equation*}
The control space is $\{u | \lvert u \rvert \leq 10\}$.

\myparagraph{System III: Third-Order Strict Feedback Form} A three-dimensional system with the form
\begin{equation}
  \dot{x_1} = e_1 x_2,~
  \dot{x_2} = e_2 x_3,~
  \dot{x_3} = e_3 x_1^2 + e_4 u,
  \end{equation}
where the states are $(x_1,x_2,x_3)$ and the state space is 
\begin{equation*}
  \mathcal{D} = \{(x_1,x_2,x_3) | \lvert x_1 \rvert \leq 1.5,
  \lvert x_2 \rvert \leq 1.5,
  \lvert x_3 \rvert \leq 2 \}.
\end{equation*}
The control space is $\{u | \lvert u \rvert \leq 30\}$.

\myparagraph{System IV: Cart Pole}  
A system with dynamics 
\begin{equation}
    \begin{aligned}
      &(M+m)\ddot{x} - ml\ddot{\theta}\cos{\theta} + ml\dot{\theta}^2\sin{\theta} = u, \\
      &ml^2\ddot{\theta} - mgl\sin{\theta} = ml\ddot{x}\cos{\theta}.
    \end{aligned}
\end{equation}
The states are $(x,\theta,\dot{x},\dot{\theta})$, where $\theta = 0$ is the upward position of the pole, and $\theta$ is positive in the counter-clockwise direction. The state space is
\begin{equation*}
  \mathcal{D} =\{(x,\theta,\dot{x},\dot{\theta}) | \lvert x \rvert \leq 1,
    \lvert \theta \rvert \leq \pi/6,
    \lvert \dot{x} \rvert \leq 1,
    \lvert \dot{\theta} \rvert \leq 1 \}.
\end{equation*}
The control space is $\{u | \lvert u \rvert \leq 30\}$. Previous approaches have found it difficult to find a controller for this system.


\subsection{Comparison with Baseline}\label{subsec:algorithm1}

We first demonstrate the efficacy and efficiency of our monotonic network in terms of training convergence time over a fixed bounded state space $\mathcal{D}$ using Algorithm~\ref{alg:lyapTraining_withSet}. As a baseline for comparison, we use a simple $\relu$ network~\cite{DaiH-RSS-21} (which requires verification of all Lyapunov conditions~\ref{eq:lyapunov_conditions_a}-\ref{eq:lyapunov_conditions_c} during training). We apply both approaches to the four systems described in Sec.~\ref{subsec:systems} and compare them in terms of training time. The network structure for the baseline was selected on a trial-and-error basis for good performance, and our network structured is chosen accordingly to use a comparable number of $\relu$ units (therefore, the same number of binary variables). We use the same network structure for the controller in each. The network structure and the total training time is presented in the first section of Table.~\ref{table:example_networks}. These results show a significant improvement in training time, ranging from a reduction of 24\% for the 4-D cart-pole system, up to 90\% for the 2D path following case.
\begin{figure}[htbp!]
  \centering
  \subfloat[Pendulum\label{fig:baseline_pendulum}]{%
  \includegraphics[width=0.48\linewidth]{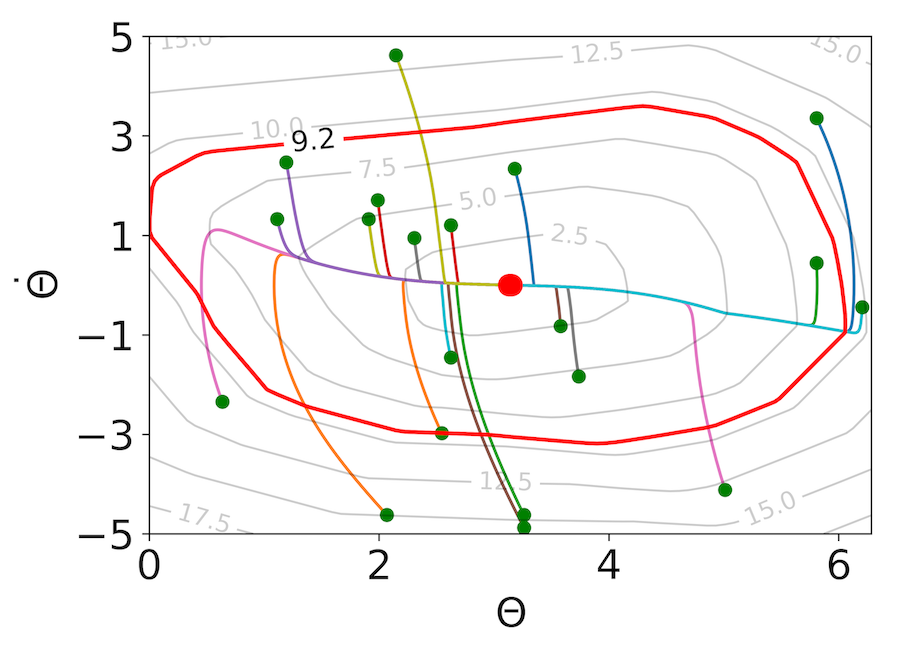}%
}
\quad
\subfloat[Path Following\label{fig:baseline_pf}]{%
\includegraphics[width=0.46\linewidth]{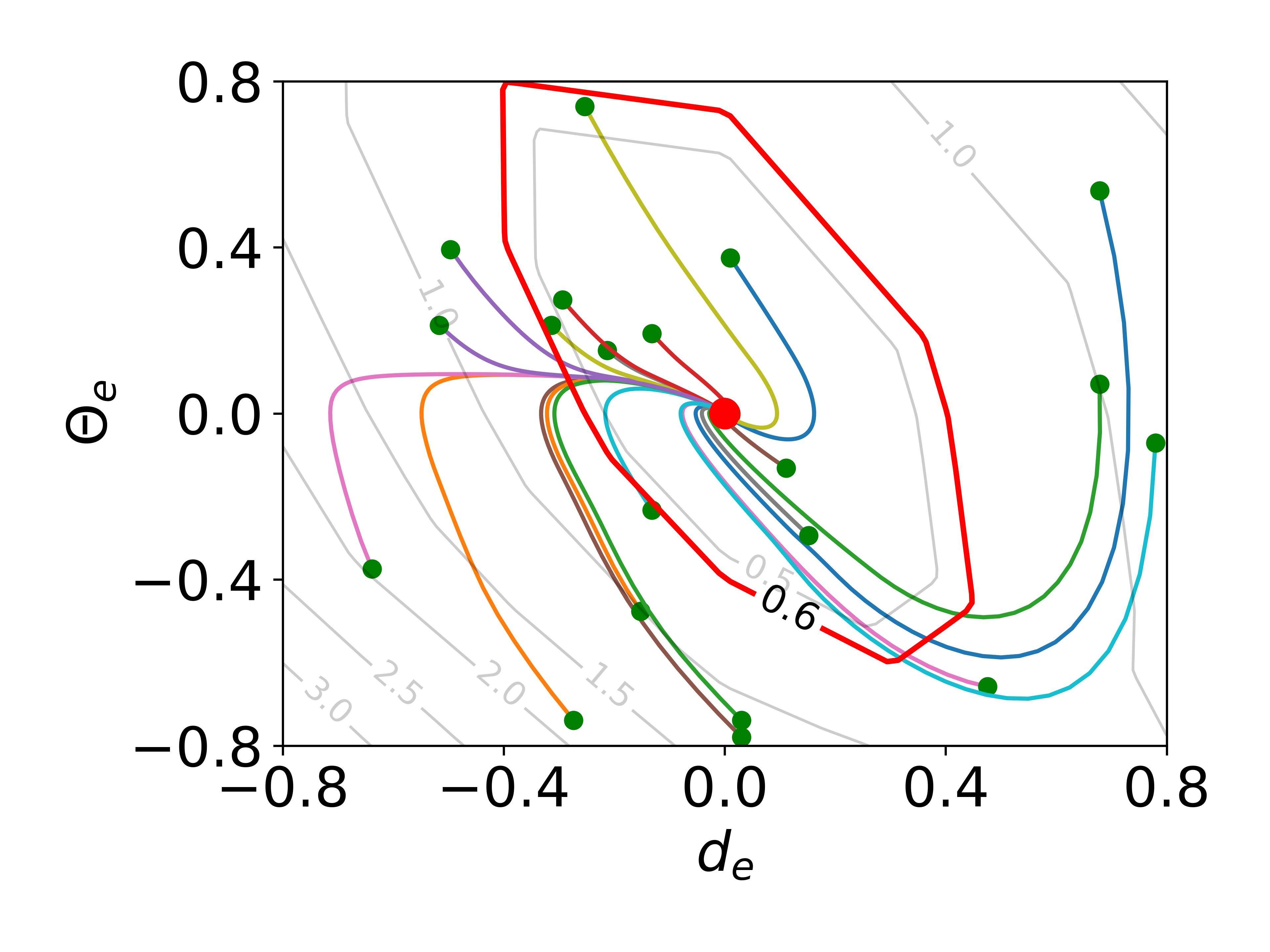}%
}
  \caption{Final Lyapunov function and example controlled trajectories after applying Algo.~\ref{algo:verfication_in_set}. (bold red curve) ROA within the bounded domain. (red dot) Equilibrium point. (green dots) Randomly selected initial conditions, and the corresponding system trajectories under the learned controller.}
    \label{fig:baseline_comparison}
\end{figure}

Fig. ~\ref{fig:baseline_comparison} shows two examples of the resulting Lyapunov function and controlled trajectories, focusing on the 2D systems for easy visualization.

\subsection{ROA Expansion}
Algo.~\ref{alg:lyapExpansion_withSet} combines learning the controller, certifying the Lyapunov function, and expanding the ROA.Fig.~\ref{fig:roa_expansion} shows results on System I (the inverted pendulum). Although the training time is longer than with Algorithm 1 (see the last column in Table 1), it can yield a significantly larger ROA. Figs~\ref{fig:roa_expansion}-\ref{fig:cart_pole_roa} show the results for the four example systems.

Fig.~\ref{fig:init_roa} shows the initial Lyapunov function, the defined sub-level set (chosen to yield a relatively small initial region to allow room for the algorithm to adjust the Lyapunov function and grow the ROA), and the inscribed square for System I. Fig.~\ref{fig:expand_roa} shows the final result, with the level set shown in bold red. Note that the algorithm allows the ROA to extend beyond the domain $\mathcal{D}$. 
To guarantee the ROA, we then find the maximal level set that lies entirely within $\mathcal{D}$ (shown in green). The performance of the resulting controller is demonstrated in Fig.~\ref{fig:expand_roa_traj}-\ref{fig:expand_roa_lyap}, which shows a collection of trajectories from randomly selected initial points (compare against the results from Algo.~\ref{alg:lyapTraining_withSet} in Fig.~\ref{fig:baseline_pendulum}) and the corresponding evolution of the Lyapunov function. Note that some of the selected initial conditions lay outside the ROA but are stabilized nevertheless.

The results of expansion for Systems II and III (Fig.~\ref{fig:3d_order}) both show a significant increase in the final ROA over the initial one. To avoid overly cluttering the image, the results for Systems III (Fig.~\ref{fig:third_order_roa}) do not include the maximal level set inside $\mathcal{D}$. 
\begin{figure}[ht]
  \centering
  \subfloat[Initial\label{fig:init_roa}]{%
  \includegraphics[width=0.43\linewidth]{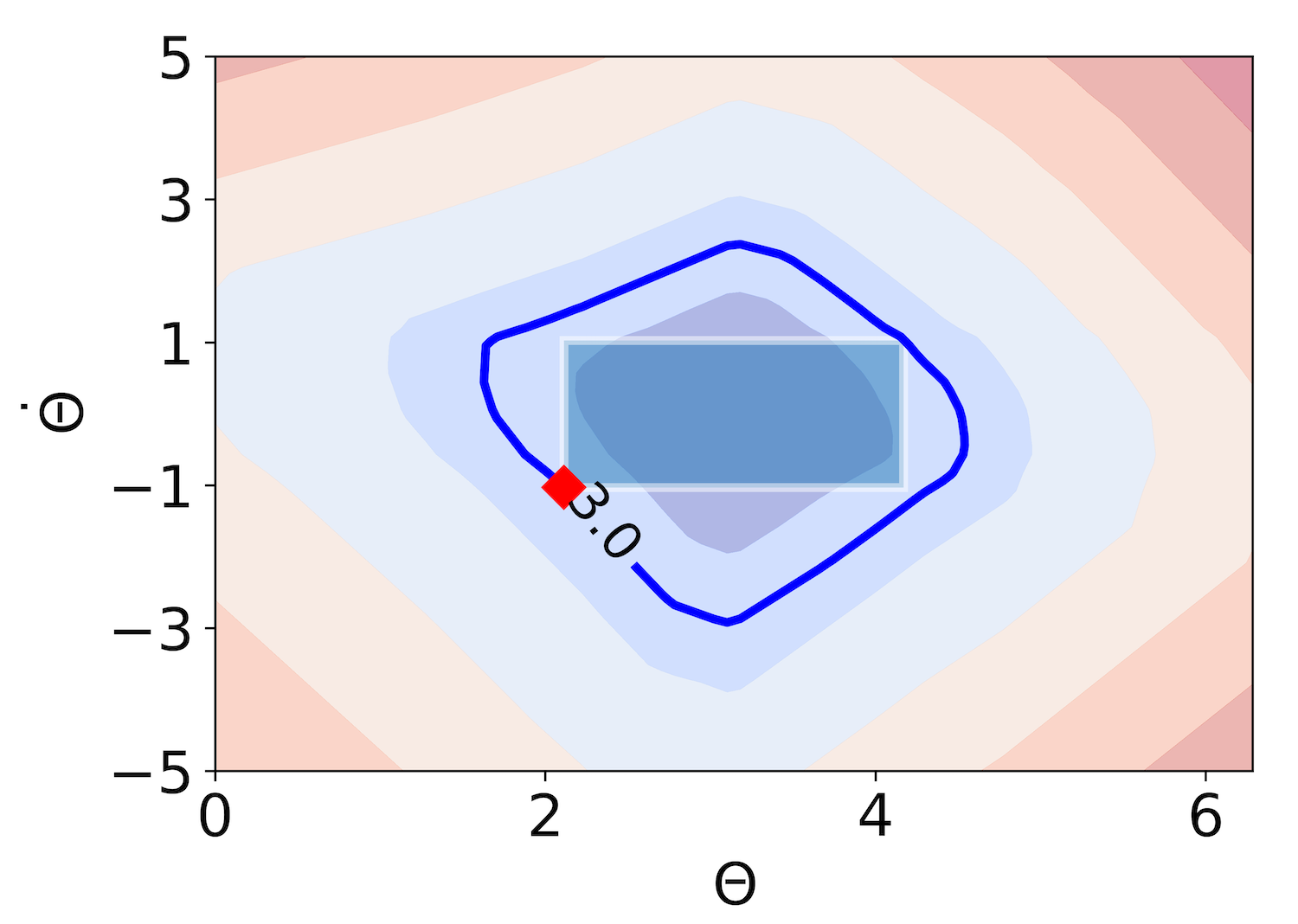}%
}
\quad
  \subfloat[Final\label{fig:expand_roa}]{%
  \includegraphics[width=0.43\linewidth]{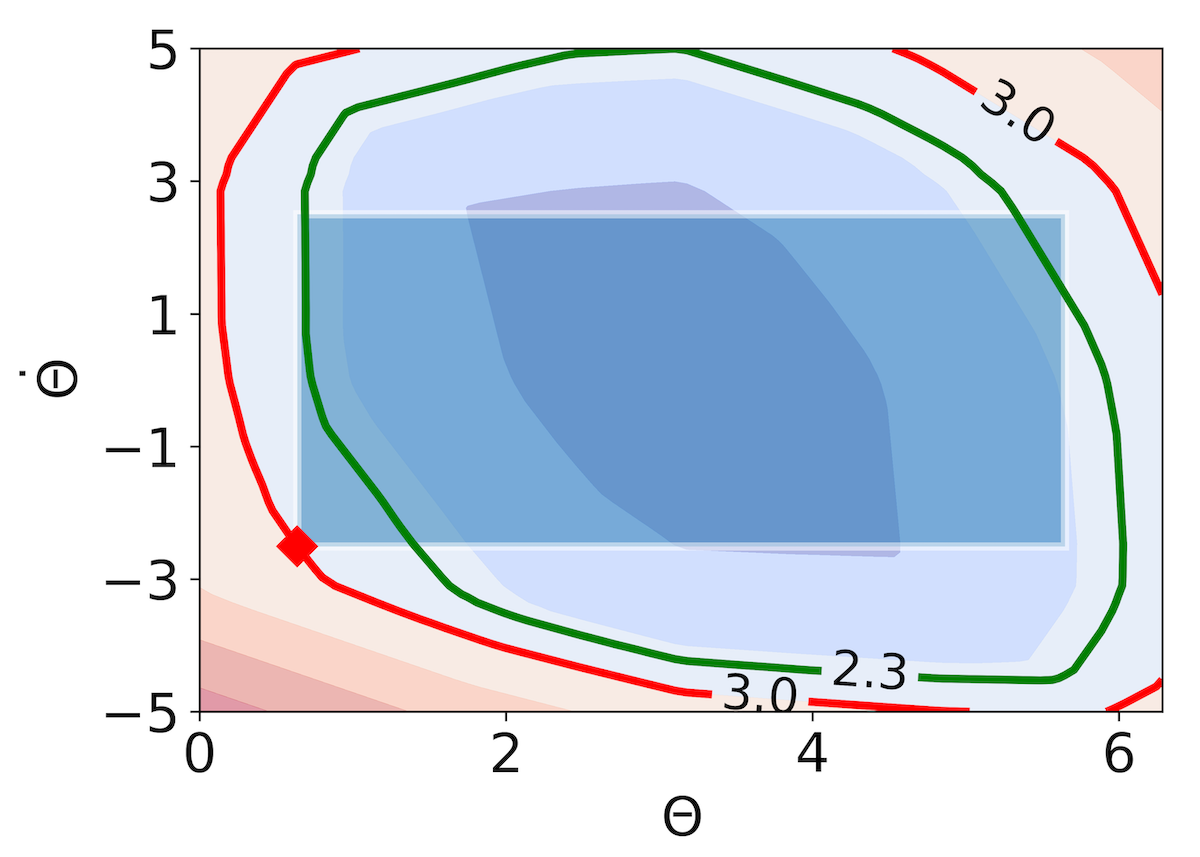}%
}
\quad
\subfloat[Trajectories\label{fig:expand_roa_traj}]{%
\includegraphics[width=0.43\linewidth]{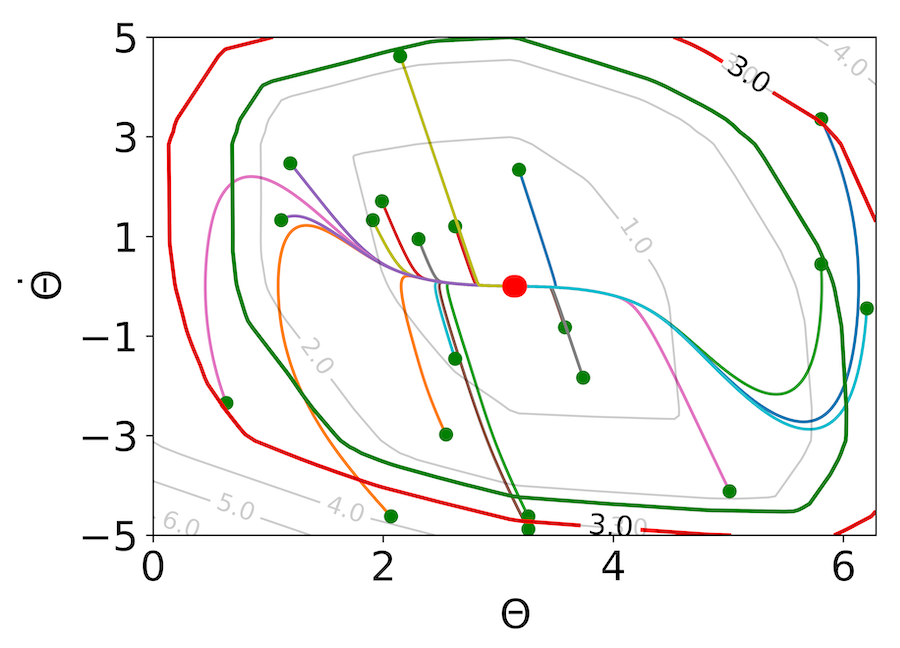}}
\quad
\subfloat[Lyapunov\label{fig:expand_roa_lyap}]{%
\includegraphics[width=0.42\linewidth]{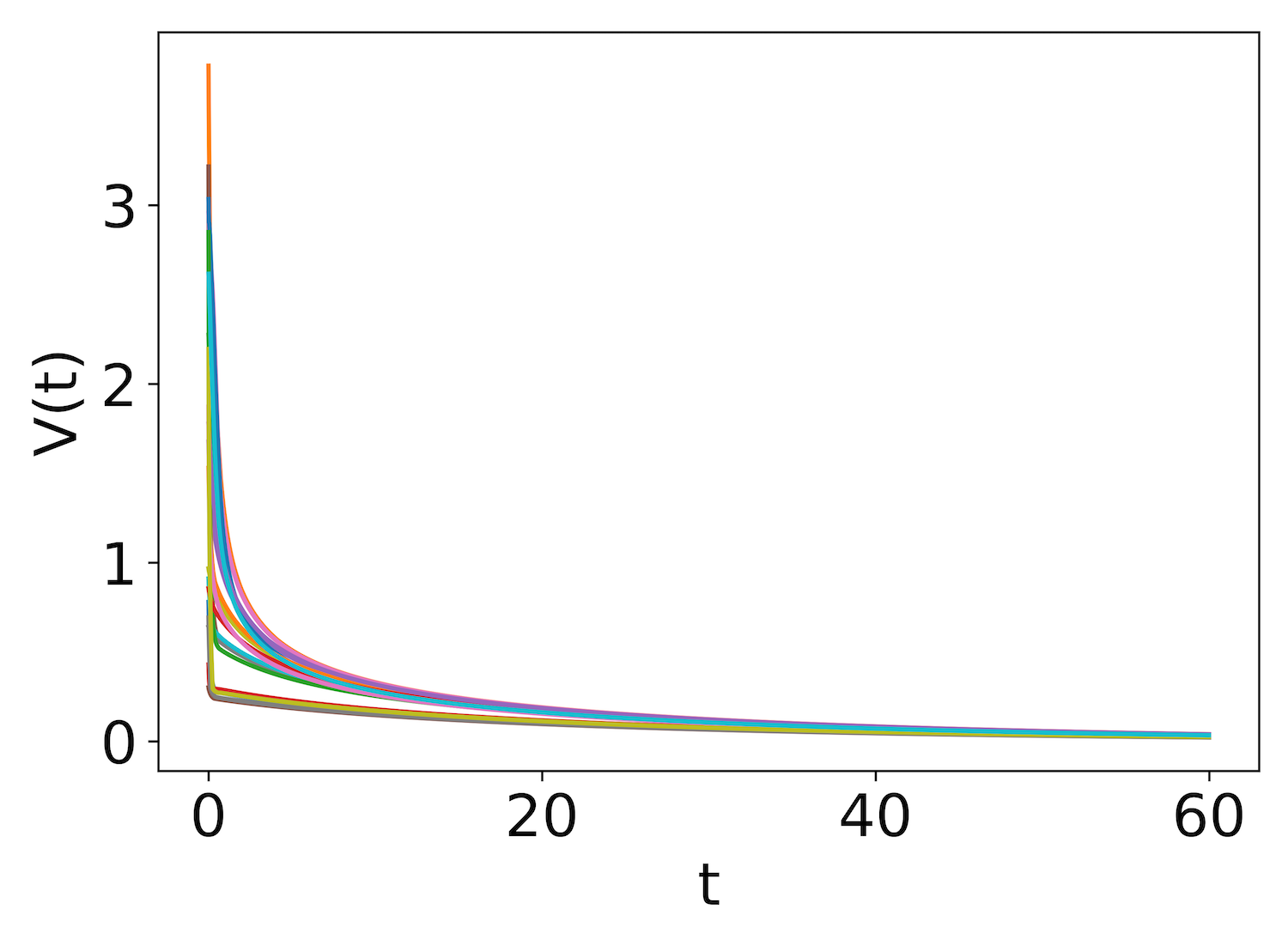}
}
\caption{Illustration of ROA in the inverted pendulum example. (a) $\mathcal{R}_{3.0}$ before expansion is highlighted in blue. (b) $\mathcal{R}_{3.0}$ after expansion is in red; the $l_{\infty}$ norm 
is shown as the inscribed blue box, and the intersection between $l_{\infty}$ norm and $\mathcal{R}_{3.0}$ is shown as a red dot. The level set touching the state space bounds is shown in green  ($\mathcal{R}_{2.3}$). (c) Trajectories. (d) Lyapunov values.}
\label{fig:roa_expansion}
\end{figure}
The results for System IV (the cart-pole) are shown in Figs.~\ref{fig:cart_pole_roa} and \ref{fig:cart_pole_traj}. Because this is a 4-D system, 
we have chosen to show four 3D slices of the ROAs, presenting clear improvement after expansion. It is also interesting to note that the final ROAs appear to be elongated in some dimensions; in future work we intend to consider expansions using inscribed zonotopes to further improve the result. The controlled evolution of the system from a randomly selected set of initial conditions shown in Fig.~\ref{fig:cart_pole_traj} clearly demonstrate that the system converges to the equilibrium point.
\begin{figure}[ht]
    \centering
    \subfloat[Path Following\label{fig:pf_roa}]{%
    \includegraphics[width=0.46\linewidth]{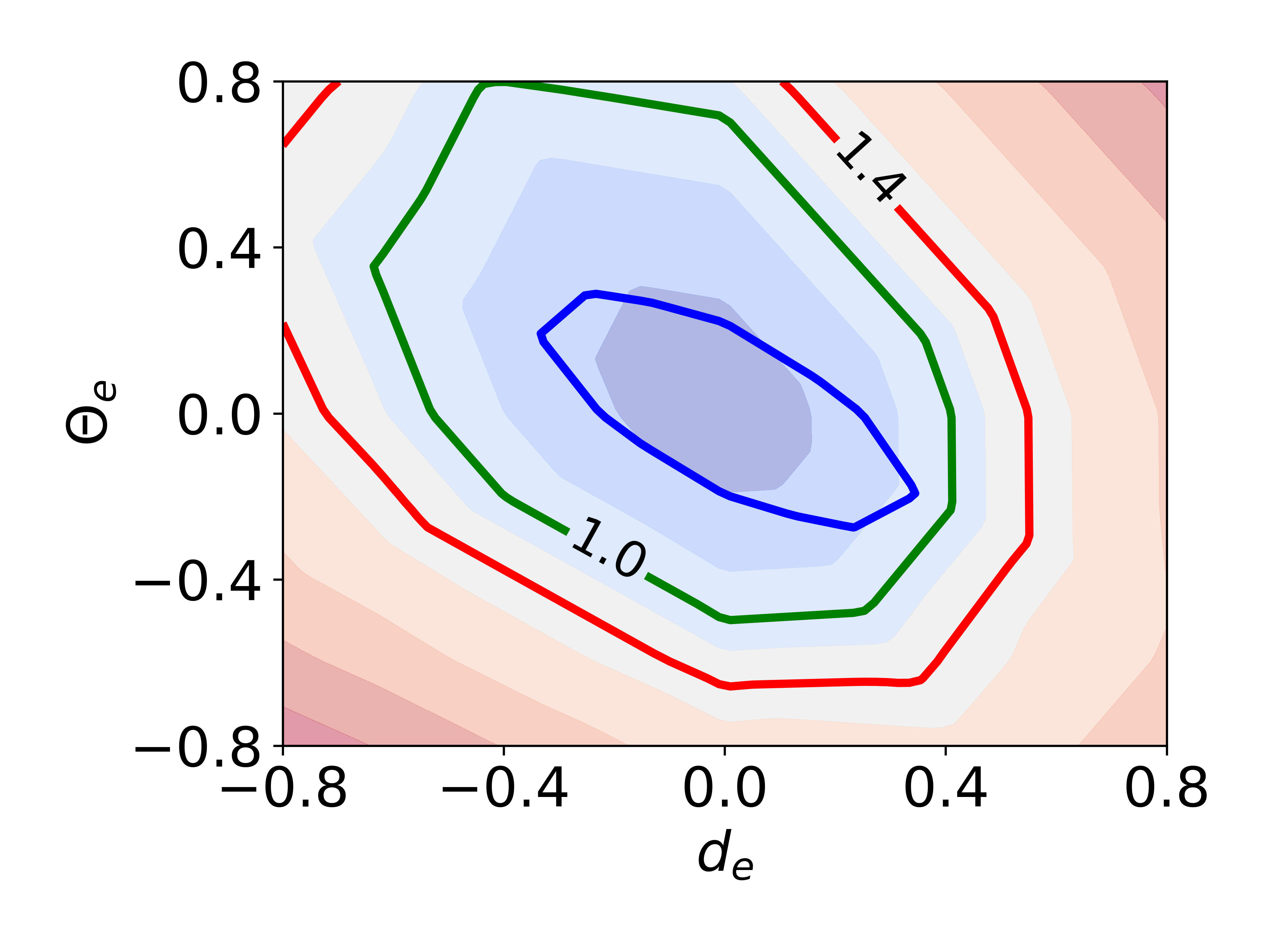}%
  }
  \quad
    \subfloat[Third Order\label{fig:third_order_roa}]{%
    \includegraphics[width=0.44\linewidth]{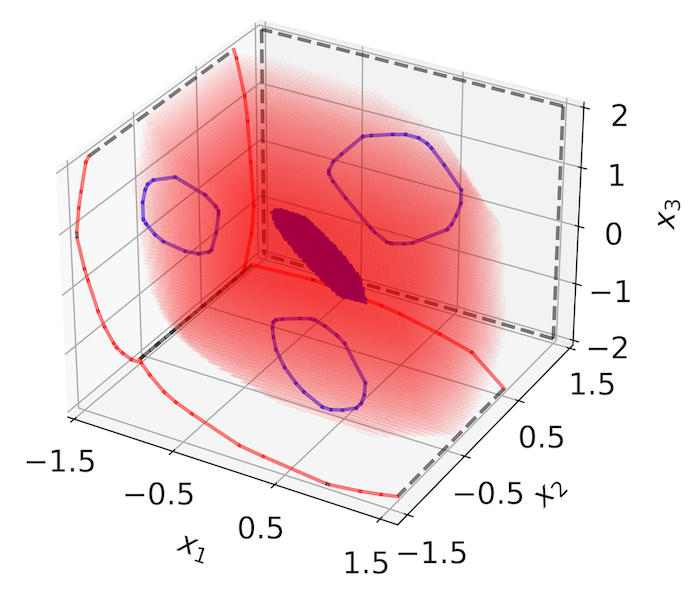}%
    }
    \caption{ROA Expansion. (blue) Initial level set. (red) Final level set. (a) System II, where (green) the maximal level set fully within the domain is also shown. (b) System III. Level sets are projected onto the corresponding planes. Dashed black line indicates the projection is outside of $\mathcal{D}$.}
      \label{fig:3d_order}
\end{figure}
\begin{figure}[ht]
  \centering
  \subfloat[$\dot{\theta}=0$]{%
  \includegraphics[width=0.43\linewidth]{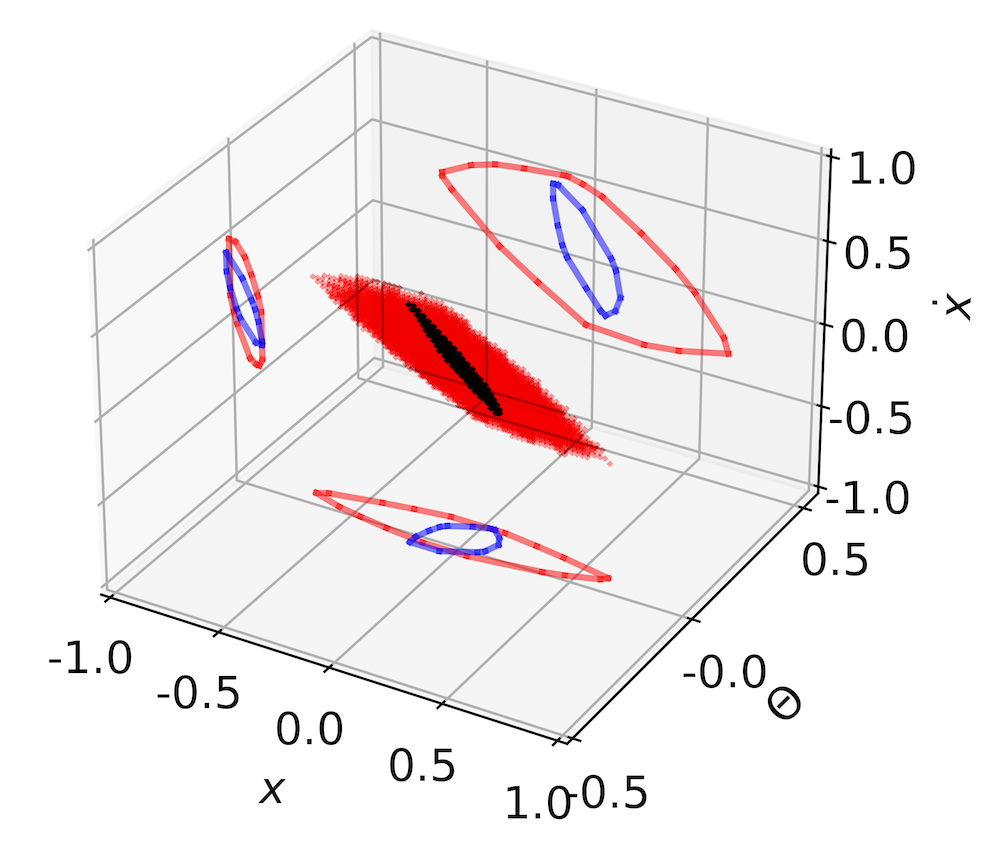}%
  }
  \subfloat[$\dot{x}=0$]{%
  \includegraphics[width=0.43\linewidth]{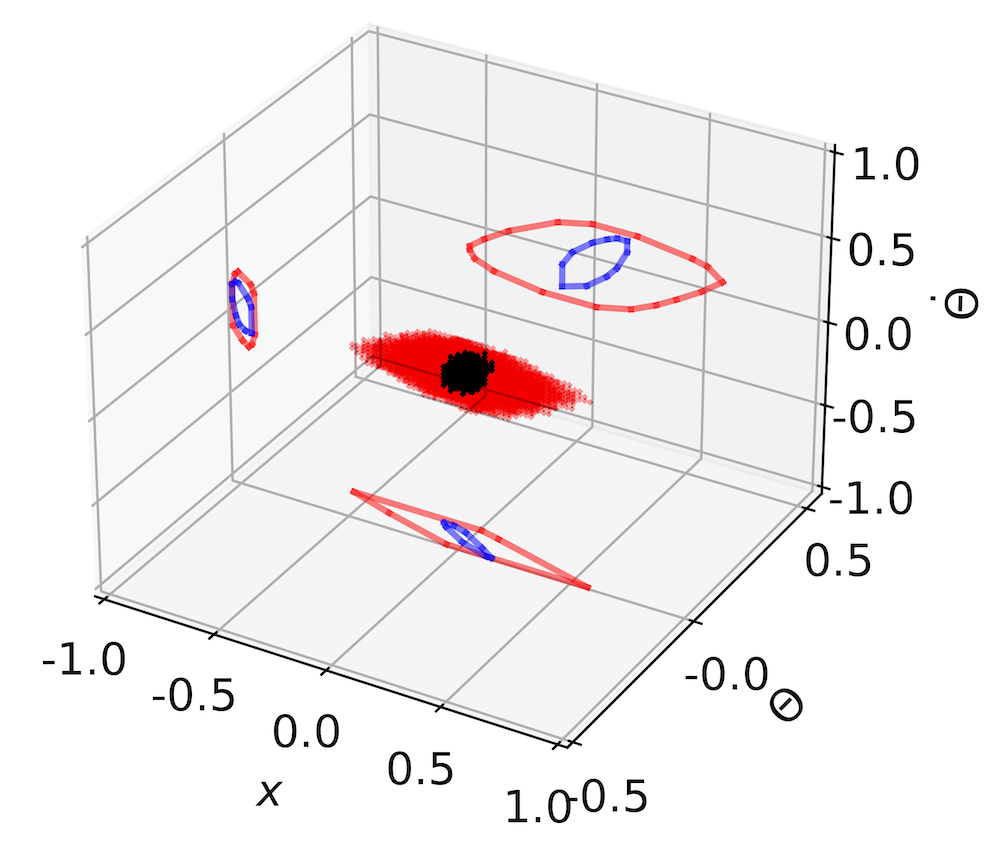}%
}
\quad
\subfloat[$x=0$]{%
\includegraphics[width=0.43\linewidth]{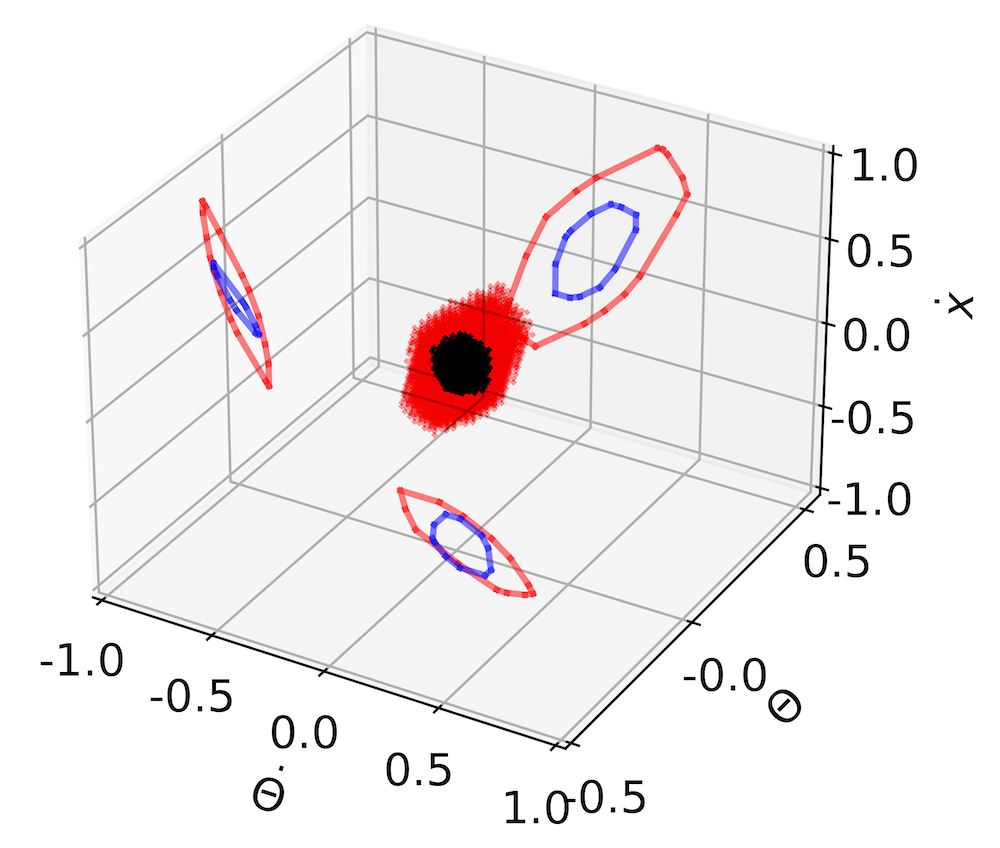}%
}
\subfloat[$\theta=0$]{%
\includegraphics[width=0.43\linewidth]{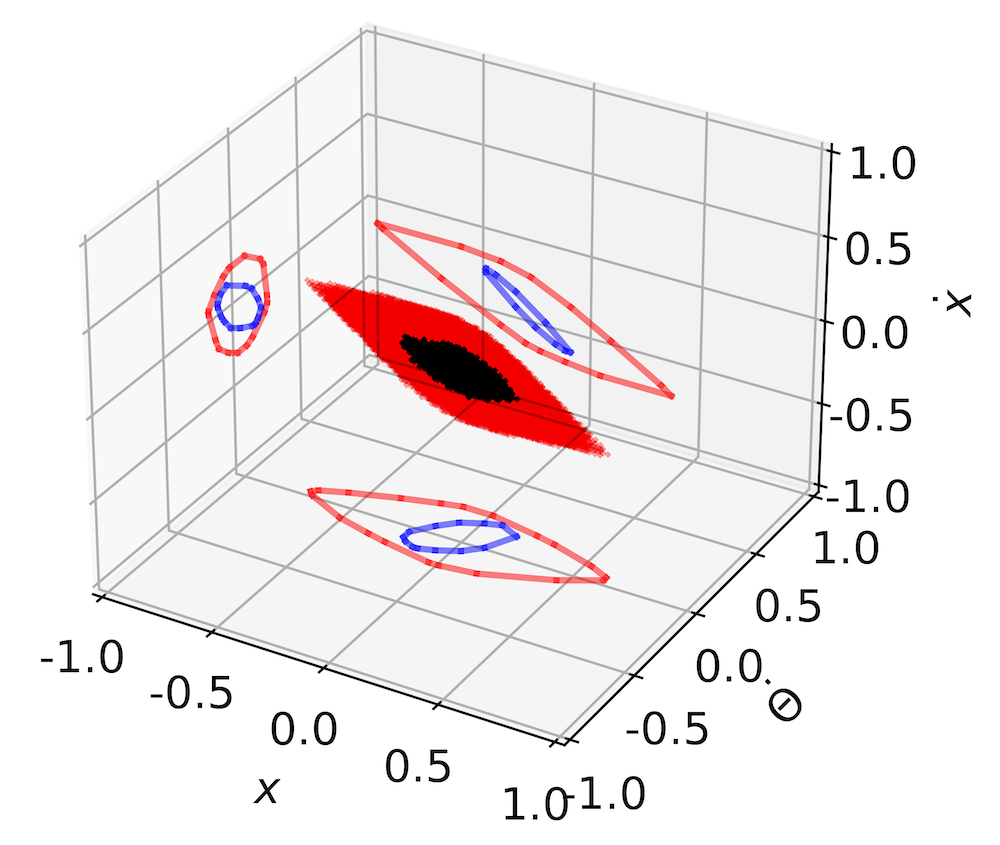}%
}
  \caption{ROA Expansion on System IV. Each plot fixes one of the states to visualize the corresponding 3D slice.}
    \label{fig:cart_pole_roa}
\end{figure}

\begin{figure}[ht]
  \centering

  \includegraphics[width=0.99\linewidth]{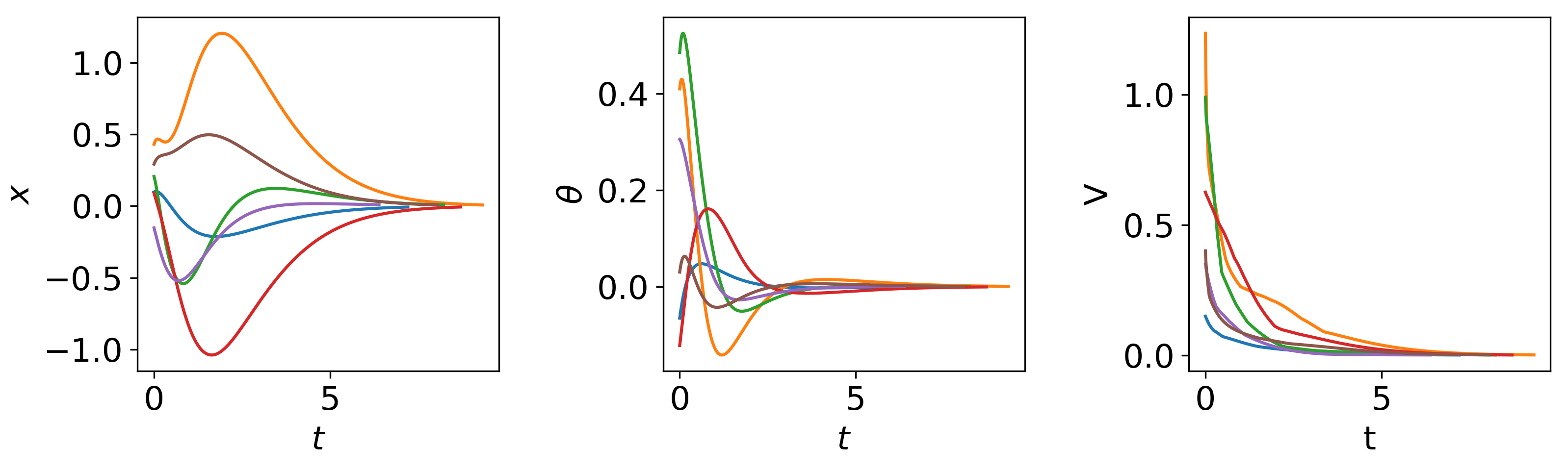}%
  \caption{Controlled trajectories on System IV.}
    \label{fig:cart_pole_traj}
\end{figure}

\section{Conclusions}
In this paper, we proposed a $\relu$ neural network that by construction satisfies the non-negativity property of a Lyapunov function. Our algorithm to learn both controller and verifying Lyapunov function for nonlinear systems requires significantly less training time than a baseline. Taking advantage of the inherent properties of our network structure, we developed an algorithm that concurrently expands the region of attraction for the system, ensuring both stability and safety over as large a domain in the state space as possible. 

In this work we used only a single layer in the Lyapunov neural network, though the proposed structure is more general and in future work we intend to explore the advantage of the more general function representation, including adding additional layers. We also intend to consider modifications to our expansion algorithm that rely on counter-example guided training, similar to~\cite{DaiH-RSS-21}, instead of the min-max approach we currently use.
Finally, our proposed method does not consider any uncertainties in the model. In the future, we plan to analyze the neural network approximations with model errors, using a robust control Lyapunov function approach~\cite{chen2021learning_roa}. 
We also plan to apply our approach in real systems, combining it with control barrier functions to avoid obstacles.

\bibliographystyle{biblio/IEEEtranNoDash}


\input{2024-ACC-LyapunovNeuralNetworks_finalSubmission.bbl}
\end{document}

%% file: preamble/common.tex
\input{preamble/commonNoTikz}
\input{preamble/graphicsTikz}

\input{preamble/robotics}


%% file: preamble/commonNoTikz.tex
\input{preamble/fixes}
\input{preamble/graphics}
\input{preamble/pagination}
\input{preamble/utilities}
\input{preamble/math}
\input{preamble/operators}
\input{preamble/notation}

\input{preamble/units}

\input{preamble/markupAndCommenting}


%% file: preamble/fixes.tex
\usepackage{fix-cm}
\usepackage{etex}

\usepackage{dblfloatfix}

\usepackage{nag}


%% file: preamble/graphics.tex
\makeatletter
\@ifpackageloaded{xcolor}{}{%
\usepackage[table,x11names,dvipsnames,svgnames]{xcolor}%
}
\makeatother

\usepackage{colortbl}

\usepackage{graphicx}
\usepackage{wrapfig}

\input{preamble/graphicsColors}


%% file: preamble/graphicsColors.tex
\definecolorset{RGB}{lyft}{}{Red,194,39,36;Sunset,202,53,33;Orange,205,68,20;Amber,200,117,42;Yellow,242,169,52;Citron,186,188,44;Lime,112,159,33;Green,56,139,31;Mint,45,118,56;Teal,52,133,135;Cyan,60,132,202;Blue,55,94,248;Indigo,64,13,247;Purple,115,42,248;Pink,176,25,145;Rose,176,32,75}

%% file: preamble/pagination.tex
\usepackage{cite}

\usepackage{microtype}

\usepackage[american]{babel}

\usepackage{array}
\usepackage{multirow}
\usepackage{booktabs}
\usepackage{makecell} 

\newcommand{\myparagraph}[1]{\textbf{\emph{#1}}.}

\ifcsname labelindent\endcsname
\let\labelindent\relax
\fi
\usepackage[inline]{enumitem}

\usepackage{subfig}

\newenvironment{lenumerate}[2][]
{\begin{enumerate}[label=(#2\arabic*),leftmargin=0.2in,itemindent=0.15in,#1]}
{\end{enumerate}}

\setlist*[enumerate,1]{label={\itshape\arabic*)}}

\makeatletter
\newcommand{\paragraphswithstop}{%
\let\copyparagraph\paragraph%
\renewcommand\paragraph[1]{\copyparagraph{##1.}}%
}
\makeatother

\usepackage[framemethod=tikz]{mdframed}

\makeatletter
\def\namedlabel#1#2{\begingroup
  #2%
  \def\@currentlabel{#2}%
  \phantomsection\label{#1}\endgroup
}
\makeatother


%% file: preamble/utilities.tex
\usepackage{suffix}

\usepackage{environ}

\makeatletter
\newsavebox{\boxifnotempty}
\newcommand{\displayifnotempty}[3]{\sbox\boxifnotempty{#2}\setbox0=\hbox{\usebox{\boxifnotempty}\unskip}%
\ifdim\wd0=0pt
\else
 #1\usebox{\boxifnotempty}#3%
\fi%
}

\newcommand{\ifempty}[2]{\setbox0=\hbox{#1\unskip}%
\ifdim\wd0=0pt%
 #2%
\fi%
}

\newcommand{\ifnotempty}[2]{\setbox0=\hbox{#1\unskip}%
\ifdim\wd0>0pt%
 #2%
\fi%
}
\makeatother

\newcommand{\switchifempty}[3]{\sbox\boxifnotempty{#1}\setbox0=\hbox{\usebox{\boxifnotempty}\unskip}%
  \ifdim\wd0=0pt{}%
  #2%
  \else{}%
  #3%
  \usebox{\boxifnotempty}%
  \fi%
}

\usepackage{algorithm}

\usepackage{scrlfile}

\makeatletter
\newcommand*\newstoreddef[1]{
  \BeforeClosingMainAux{%
    \immediate\write\@auxout{%
      \string\restoredef{#1}{\csname #1\endcsname}%
    }%
  }%
}
\newcommand*{\restoredef}[2]{
  \expandafter\gdef\csname stored@#1\endcsname{#2}%
}
\newcommand*{\storeddef}[1]{
  \@ifundefined{stored@#1}{0}{\csname stored@#1\endcsname}%
}
\makeatother

\usepackage{pageslts}
\NewEnviron{tee}{\BODY\typeout{Marker Tee [start] ^^J \BODY ^^JMaker Tee [end]}}


%% file: preamble/operators.tex
\newcommand{\real}[1]{\mathbb{R}^{#1}{}}

\newcommand{\naturals}[1]{\mathbb{N}^{#1}{}}

\newcommand{\transpose}{^\mathrm{T}}

\newcommand{\inverse}{^{-1}}

\newcommand{\vct}[1]{\mathbf{#1}}

\DeclareMathOperator{\cvxhull}{co}
\DeclareMathOperator{\relu}{ReLU}



%% file: preamble/notation.tex
\providecommand{\cK}{\mathcal{K}}

\providecommand{\cM}{\mathcal{M}}


%% file: preamble/units.tex
\usepackage{units}


%% file: preamble/markupAndCommenting.tex
  \newcommand{\newcolorlabel}[2]{%
  \expandafter\newcommand\csname #1\endcsname[1]{%
    \tikz[baseline]{\node[text=white,fill=#2,anchor=base,text height=1.3ex,text depth=0.1ex,font=\sffamily\bfseries]{##1}}}%
}

\newcommand{\newcommenter}[2]{%
  \expandafter\newcommand\csname #1\endcsname[1]{%
    \fcolorbox{#2}{#2}{\color{white}\textsf{\textbf{#1}}}
    {\color{#2}##1}}%
  \expandafter\newcommand\csname at#1\endcsname{%
    \fcolorbox{#2}{#2}{\color{white}\textsf{\textbf{@#1}}}
    {\color{#2}}}%
  \expandafter\newcommand\csname #1cite\endcsname[1]{%
    \csname #1\endcsname {[##1]}
  }%
  \expandafter\newcommand\csname #1ref\endcsname[1]{%
    \csname #1\endcsname {$\blacktriangleright$##1}
  }%
  \expandafter\newcommand\csname #1hl\endcsname[2]{%
    \colorbox{#2}{\color{white}\textsf{\textbf{#1}}}\sethlcolor{Azure2}\hl{##2}~%
    \expandafter\ifx\csname commentarrow\endcsname\relax$\leftarrow$\else \commentarrow[#2]\fi~%
    {\color{#2}##1}}%
  \expandafter\newcommand\csname #1st\endcsname[2]{%
    \colorbox{#2}{\color{white}\textsf{\textbf{#1}}}\sout{##2}~%
    \expandafter\ifx\csname commentarrow\endcsname\relax$\leftarrow$\else \commentarrow[#2]\fi~%
    {\color{#2}##1}}%
}
\newcommenter{TODO}{DodgerBlue1}
\newcommenter{rtron}{Green3}

\usepackage{comment}

\usepackage{pdfcomment}

\usepackage{soul}

\usepackage[normalem]{ulem}

\usepackage{csquotes}


%% file: preamble/graphicsTikz.tex
\usepackage{tikz}
\usetikzlibrary{calc}
\usetikzlibrary{matrix}
\usetikzlibrary{chains,scopes}
\usetikzlibrary{shapes.geometric}
\usetikzlibrary{arrows.meta}
\usetikzlibrary{decorations.markings}
\usetikzlibrary{decorations.pathreplacing}
\usetikzlibrary{backgrounds}

\tikzset{
  dim above/.style={to path={\pgfextra{
        \pgfinterruptpath
        \draw[>=latex,|->|] let
        \p1=($(\tikztostart)!1.5em!90:(\tikztotarget)$),
        \p2=($(\tikztotarget)!1.5em!-90:(\tikztostart)$)
        in(\p1) -- (\p2) node[pos=.5,sloped,above]{#1};
        \endpgfinterruptpath
      }
    }
  },
  dim double above/.style={to path={\pgfextra{
        \pgfinterruptpath
        \draw[>=latex,|->|] let
        \p1=($(\tikztostart)!3em!90:(\tikztotarget)$),
        \p2=($(\tikztotarget)!3em!-90:(\tikztostart)$)
        in(\p1) -- (\p2) node[pos=.5,sloped,above]{#1};
        \endpgfinterruptpath
      }
    }
  },
  dim below/.style={to path={\pgfextra{
        \pgfinterruptpath
        \draw[>=latex,|->|] let 
        \p1=($(\tikztostart)!-1em!-90:(\tikztotarget)$),
        \p2=($(\tikztotarget)!-1em!90:(\tikztostart)$)
        in (\p1) -- (\p2) node[pos=.5,sloped,below]{#1};
        \endpgfinterruptpath
      }
    }
  },
}

\tikzset{
    right angle quadrant/.code={
        \pgfmathsetmacro\quadranta{{1,1,-1,-1}[#1-1]}     
        \pgfmathsetmacro\quadrantb{{1,-1,-1,1}[#1-1]}},
    right angle quadrant=1, 
    right angle length/.code={\def\rightanglelength{#1}},   
    right angle length=2ex, 
    right angle symbol/.style n args={3}{
        insert path={
            let \p0 = ($(#1)!(#3)!(#2)$) in     
                let \p1 = ($(\p0)!\quadranta*\rightanglelength!(#3)$), 
                \p2 = ($(\p0)!\quadrantb*\rightanglelength!(#2)$) in 
                let \p3 = ($(\p1)+(\p2)-(\p0)$) in  
            (\p1) -- (\p3) -- (\p2)
        }
    }
}

\newcommand{\pgfextractangle}[3]{%
    \pgfmathanglebetweenpoints{\pgfpointanchor{#2}{center}}
                              {\pgfpointanchor{#3}{center}}
    \global\let#1\pgfmathresult  
}

\usetikzlibrary{shapes.arrows}
\newcommand{\commentarrow}[1][Azure4]{\tikz[baseline=-3pt]{\node[shape border uses incircle, fill=#1,rotate=180,single arrow, inner sep=1pt, minimum size=6pt, single arrow head extend=2pt]{};}}


%% file: preamble/robotics.tex
\tikzset{ax/.style={-latex,line width=2pt}}

\tikzset{camera/.style={fill=Sienna1,fill opacity=0.5},%
image plane/.style={draw=RoyalBlue3,line width=2pt}}

